\documentclass[pra,longbibliography,twocolumn,showpacs,superscriptaddress,notitlepage]{revtex4-2}
\usepackage{amsmath}
\usepackage{amssymb,bm}
\usepackage{amsthm}

\usepackage{qcircuit}

\newtheorem{innercustomgeneric}{\customgenericname}
\providecommand{\customgenericname}{}
\newcommand{\newcustomtheorem}[2]{%
  \newenvironment{#1}[1]
  {%
   \renewcommand\customgenericname{#2}%
   \renewcommand\theinnercustomgeneric{##1}%
   \innercustomgeneric
  }
  {\endinnercustomgeneric}
}

\newcustomtheorem{customthm}{Theorem}
\newcustomtheorem{definitionBob}{Definition}

\usepackage{color,dsfont} 
\usepackage{graphicx}
\usepackage{subcaption}
\usepackage{ragged2e}
\DeclareCaptionJustification{justified}{\justifying}
\usepackage[colorlinks=true, hyperindex, breaklinks, linkcolor=blue, urlcolor=blue, citecolor=blue]{hyperref} % Physical Review style
\usepackage[normalem]{ulem}
\usepackage[capitalise]{cleveref}
\usepackage{mathrsfs}

\usepackage{mathtools}
\usepackage{float}
\usepackage{verbatim}
\usepackage{latexsym}
\usepackage{amsmath}
\usepackage{amssymb}
\usepackage{setspace}
\usepackage{amsfonts}
\usepackage{stmaryrd}
\usepackage{xcolor}
\usepackage{enumitem}
\usepackage{hhline}
\usepackage[normalem]{ulem}

 \newtheorem{lemma}{Lemma}
 
 \newtheorem{definition}{Definition}

\newcommand{\ket}[1]{|#1\rangle}  % braket
\newcommand{\codepar}[1]{\ensuremath{[\![#1]\!]}}

\definecolor{azure}{rgb}{0.0, 0.5, 1.0}
\definecolor{blue-green}{rgb}{0.0, 0.67, 0.57}
\definecolor{neoncarrot}{rgb}{1.0, 0.64, 0.26}

\crefname{equation}{Eq.\!}{Eqs.\!}
\crefname{figure}{Fig.\!}{Figs.\!}
\usepackage{multirow}
\mathchardef\mhyphen="2D

\begin{document}

% \title{The effect of possible fault locations on the performance of planar color codes of distance 9}
% \title{Code capacity performance might not predict circuit level performance for d=9 color codes}
% \title{The code you thought was the best in code capacity might be the worse choice in circuit-level}
% Preserving the distance of self-orthogonal CSS codes with only two extra qubits
% Lookup table based soft decoding of the concatenated Steane code
%Preserving the distance of the concatenated Steane code with only two extra qubits
%Low-ancilla distance-preserving decoder for the 49-qubit concatenated Steane code
\title{Concatenated Steane code with single-flag syndrome checks}

\author{Balint Pato}
\email{balint.pato@duke.edu}
\thanks{contributed equally}
\affiliation{
	Duke Quantum Center, Duke University, Durham, NC 27701, USA
}
\affiliation{
	Department of Electrical and Computer Engineering, Duke University, Durham, NC 27708, USA
}
\author{Theerapat Tansuwannont}
\email{t.tansuwannont.qiqb@osaka-u.ac.jp}
\thanks{contributed equally}
\thanks{present address: Center for Quantum Information and Quantum Biology, Osaka University, Toyonaka, Osaka 560-0043, Japan.}
\affiliation{
	Duke Quantum Center, Duke University, Durham, NC 27701, USA
}
\affiliation{
	Department of Electrical and Computer Engineering, Duke University, Durham, NC 27708, USA
}
\author{Kenneth R. Brown}
\email{ken.brown@duke.edu}
\affiliation{
	Duke Quantum Center, Duke University, Durham, NC 27701, USA
}
\affiliation{
	Department of Electrical and Computer Engineering, Duke University, Durham, NC 27708, USA
}
\affiliation{
	Department of Physics, Duke University, Durham, NC 27708, USA
}
\affiliation{
	Department of Chemistry, Duke University, Durham, NC 27708, USA
}

\begin{abstract}
A fault-tolerant error correction (FTEC) protocol with a high error suppression rate and low overhead is very desirable for the near-term implementation of quantum computers. In this work, we develop a distance-preserving flag FTEC protocol for the \codepar{49,1,9} concatenated Steane code, which requires only two ancilla qubits per generator and can be implemented on a planar layout. We generalize the weight-parity error correction (WPEC) technique from \cite{TL21} and find a gate ordering of flag circuits for the concatenated Steane code, which makes syndrome extraction with two ancilla qubits per generator possible. The FTEC protocol is constructed using the optimization tools for flag FTEC developed in \cite{PTHB23} and is simulated under the circuit-level noise model without idling noise. Our simulations give a pseudothreshold of $1.64 \times 10^{-3}$ for the \codepar{49,1,9} concatenated Steane code, which is better than a pseudothreshold of $1.43 \times 10^{-3}$ for the \codepar{61,1,9} 6.6.6 color code simulated under the same settings. This is in contrast to the code capacity model where the \codepar{61,1,9} code performs better.
\end{abstract}

\maketitle

\section{Introduction}

\begingroup

A quantum error-correcting code \cite{Shor96} protects quantum information from local noise by encoding logical operators into non-local operators on a larger Hilbert space. In the case of stabilizer codes \cite{Gottesman97}, correcting errors is possible by measuring and decoding the syndrome, which tells us the eigenvalues of the stabilizer generators without destroying the encoded logical information. The scheme offers protection of the logical information at the cost of using multiple physical qubits per logical qubit for encoding, and a single ancilla qubit per stabilizer generator to extract each stabilizer generator's syndrome bit. 

In a realistic setting where the gate, preparation, and measurement operations are imperfect, a fault-tolerant error correction (FTEC) protocol such as the ones proposed by Shor \cite{Shor96}, Steane \cite{Steane97}, and Knill \cite{KL97} is required to curb the propagation of errors during the syndrome extraction in each round of error correction. Turning a non-fault-tolerant protocol into a fault-tolerant one might increase the total number of ancillary qubits \cite{Shor96, Steane97,KL97,CR17a,CR20,CKYZ20}, increase the depth of the syndrome extraction circuit due to repeated syndrome measurements \cite{Shor96}, or decrease the \textit{effective distance} (the minimum number of faults required to create a logical error) of the FTEC protocol \cite{BKS21, LAR11, Delfosse_2014, GJ23}. It is of great interest to reduce all of these types of overheads for various codes. Moreover, as near-term quantum platforms are already capable of testing from dozens up to hundreds of physical qubits  \cite{Morvan_2023,Kim_2023,Wang_2023,Bluvstein_2024}, it is valuable to explore low-overhead, distance-preserving protocols that are optimized for smaller codes. 

Code concatenation is a technique to construct families of codes using multiple levels of encoding; logical qubits of a lower-level code are encoded into logical qubits of a higher-level code. The earliest threshold theorems for fault-tolerant quantum computation \cite{AB08, Preskill98,KLZ96,AGP06} prove that as long as physical errors are below the fault-tolerant accuracy threshold, the logical error rate can be suppressed exponentially in the number of concatenated layers. Two big challenges in concatenated code families are decoding and complicated ancilla qubits. 

Hard decoding is one option for decoding concatenated codes, as presented by Aliferis, Gottesman, and Preskill \cite{AGP06}. Decoding and recovery are done layer-by-layer, making ``hard decisions'' about possible errors at each level of concatenation. Hard decoding can lead to a decrease in the effective distance even under the code capacity noise model. In contrast, optimal decoding of concatenated codes under the code capacity noise model was found by Poulin \cite{Poulin_2006} using ``soft decoding''. To our knowledge, Poulin's soft decoder has not yet been extended to circuit-level noise models on concatenated codes.

At higher levels of concatenation, besides the data qubits, the ancilla qubits are typically encoded in the code of the level below. This increases the qubit overhead significantly. In recent years, great progress has been made in creating lightweight ancillary structures using flag qubits to measure error syndromes and assist decoding \cite{CR17a,CR20,CB18,TCL20,CKYZ20,CR17b,CC19}.

%Recently, optimization tools for small self-orthogonal Calberbank-Shor-Steane (CSS) codes have been developed in \cite{PTHB23}, resulting in distance-preserving, low-overhead syndrome extraction with only one syndrome ancilla and one flag qubit per stabilizer generator.

In this work, we explore the performance of the \codepar{49,1,9} concatenated Steane code. We make use of the optimization tools for small self-orthogonal Calberbank-Shor-Steane (CSS) codes recently developed in \cite{PTHB23}, which provides a distance-preserving, low-overhead syndrome extraction scheme with only one syndrome ancilla qubit and one flag qubit per stabilizer generator. Our contributions are as follows: 
%In this work, we make use of the tools in \cite{PTHB23} and explore the performance of the \codepar{49,1,9} level-2 concatenated Steane code. 
(1) We improve upon the state-of-the-art result of \cite{TL21}, which proposed an FTEC protocol capable of correcting fault combinations from up to 3 faults using only two ancilla qubits. With our careful CNOT ordering, we achieve a distance-preserving protocol that can correct up to 4 faults. We also demonstrate a planar structure of the \codepar{49,1,9} concatenated Steane code. (2) We compare the performance of the \codepar{49,1,9} concatenated Steane code to the \codepar{61,1,9} 6.6.6 color code \cite{BM06}. As reported by Sabo, Aloshius, and Brown \cite{SAB22}, the \codepar{61,1,9} 6.6.6 color code slightly outperforms the \codepar{49,1,9} concatenated Steane code in exchange for a higher qubit overhead under the code capacity error model. We reproduce this result and establish that the lookup table-based (LUT) decoder is equivalent to the trellis-based decoding \cite{SAB22}. Under circuit-level noise, we find that the \codepar{49,1,9} concatenated Steane code has a pseudothreshold of $1.64 \times 10^{-3}$ which, surprisingly, slightly outperforms the \codepar{61,1,9} 6.6.6 color code that has a pseudothreshold of $1.43 \times 10^{-3}$. (3) We show that while the \codepar{49,1,9} 4.8.8 color code has the same number of qubits as the \codepar{49,1,9} concatenated Steane code, the concatenated Steane code outperforms the 4.8.8 color code when we look at the number of extra qubits required. This is because it is not possible to construct a distance-preserving FTEC with only a single flag qubit per generator for the 4.8.8 color codes, as shown in \cref{app:488nogo}. With a pseudothreshold above $10^{-3}$, the 2D-embeddable, distance-preserving FTEC protocol for the \codepar{49,1,9} concatenated Steane code thus provides an intriguing experimental target for near-term devices. 

The paper is organized as follows: in \cref{sec:QECC}, we introduce the planar structure of the \codepar{49,1,9} concatenated Steane code in detail. The noise models and the definition of fault tolerance are reviewed in \cref{sec:noise_model}, followed by our CNOT ordering and the decoders for the different noise models in \cref{sec:find_CNOT_ordering}. Numerical results are explained in \cref{sec:numerics}, which we discuss in \cref{sec:discussion} and derive our conclusions in \cref{sec:conclusion}.

\endgroup

\section{The \codepar{49,1,9} concatenated Steane code}
\label{sec:QECC}

%\HR{An \codepar{n,k,d} stabilizer code \cite{Gottesman97} encodes $k$ logical qubits using $n$ physical qubits and can correct up to $\tau=(d-1)/2$ errors, where $d$ is the code distance. A stabilizer code is usually described by its corresponding \emph{stabilizer group}, an Abelian group generated by $r=n-k$ commuting Pauli operators whose elements are called stabilizers. The coding subspace is the simultaneous +1 eigenspace of all elements in the stabilizer group. (need paraphrase)}

An \codepar{n,k,d} stabilizer code \cite{Gottesman97} is a code that uses $n$ data qubits to represent $k$ logical qubits and has distance $d$. The code can correct any error acting nontrivially on up to $\tau=\lfloor(d-1)/2\rfloor$ data qubits. A stabilizer code can be described by its corresponding \emph{stabilizer group}, an Abelian group that is generated by $r\equiv n-k$ commuting Pauli operators and does not contain $-I$. The elements of the stabilizer group are called stabilizers, and the codespace is defined by the simultaneous +1 eigenspace of all stabilizers.

In this work, we focus on the \codepar{49,1,9} concatenated Steane code \cite{Steane96b} which is a Calderbank-Shor-Steane (CSS) code \cite{CS96,Steane96b}, a stabilizer code for which the stabilizer generators can be chosen to be purely $X$ or purely $Z$ type. %The code definition and a possible planar qubit layout of the code is described below.
Before we describe the \codepar{49,1,9} concatenated Steane code, let us first consider the \codepar{7,1,3} Steane code. The code can be described by the following stabilizer generators:

\begingroup
\setlength\arraycolsep{1pt}	
\begin{equation}
	\begin{matrix}
		g^x_1: &X &I &I &X &I &X &X, & \quad & g^z_1: &Z &I &I &Z &I &Z &Z,\\
		g^x_2: &I &X &I &X &X &I &X, & \quad & g^z_2: &I &Z &I &Z &Z &I &Z,\\
		g^x_3: &I &I &X &I &X &X &X, & \quad & g^z_3: &I &I &Z &I &Z &Z &Z.
	\end{matrix}
\end{equation}%
\endgroup

Logical $X$ and logical $Z$ operators of the \codepar{7,1,3} code are of the form $\tilde{X}=X^{\otimes 7}M$ and $\tilde{Z}=Z^{\otimes 7}N$, where $M$ and $N$ are some stabilizers of the \codepar{7,1,3} code. One can verify that the minimum weight of a logical operator is 3. The data qubits of this code can be arranged on a plane as illustrated in \cref{fig: stl2}\hyperlink{target:stl2}{a}.

\begin{figure}[tbp]
\centering
\hypertarget{target:stl2}{}
\includegraphics[width=0.45\textwidth]{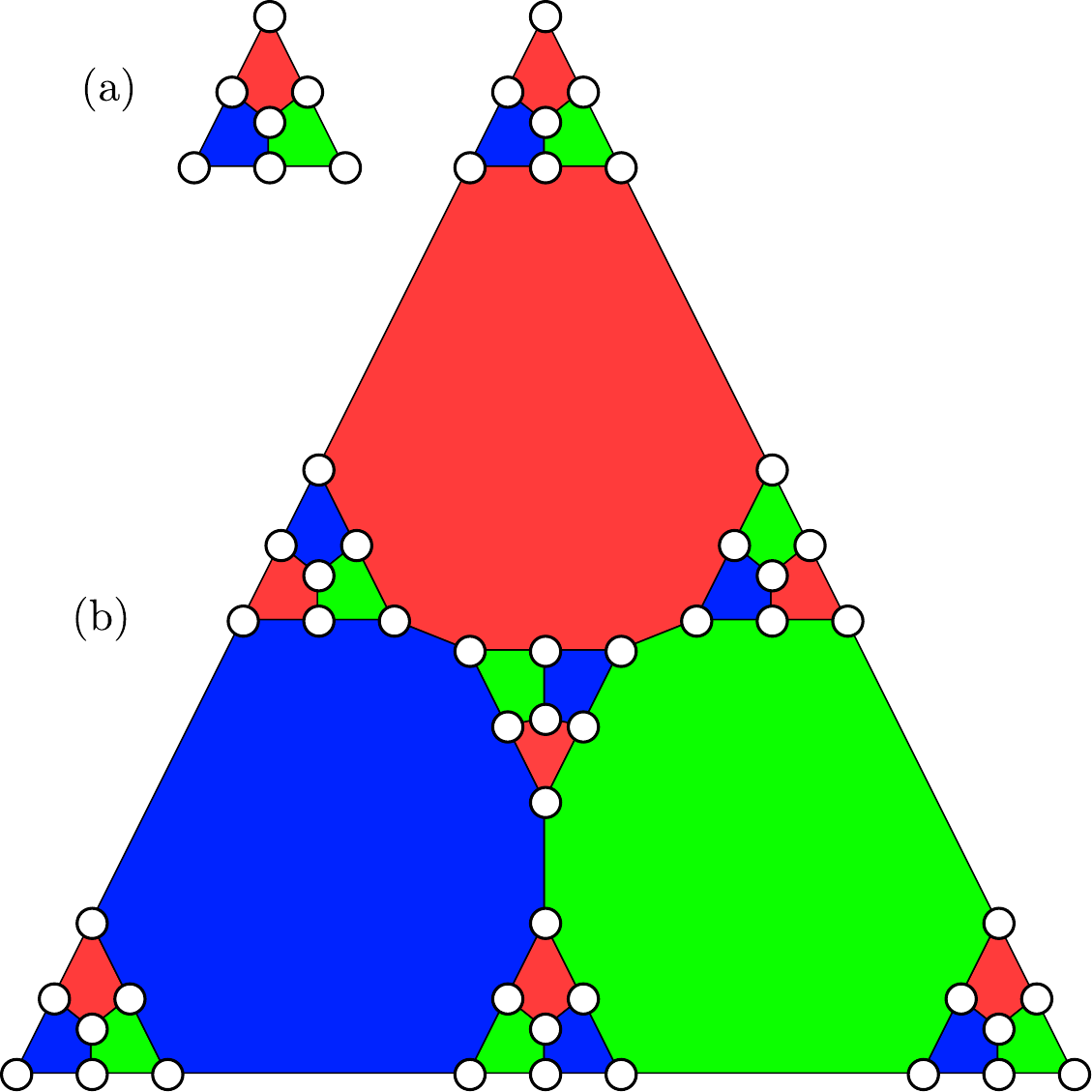}
\caption{(a) The \codepar{7,1,3} Steane code. (b) The \codepar{49,1,9} concatenated Steane code.}
\label{fig: stl2}
\end{figure}

%\begin{figure}[tbp]
%	\centering
%	\includegraphics[width=0.42\textwidth]{Fig1}
%	%\captionsetup{justification=centering}
%	\caption{Shor syndrome extraction circuit for measuring a stabilizer generator of the form $M=P_1\otimes P_2 \otimes P_3 \otimes P_4$. The ancilla qubits are initially prepared in the cat state $\frac{1}{\sqrt{2}}(|0000\rangle+|1111\rangle)$ and measured in the $Z$ basis at the end. Even and odd parities of the measurement results correspond to $+1$ and $-1$ eigenvalues of $M$.}
%	\label{fig:Shor}
%\end{figure}

The \codepar{49,1,9} concatenated Steane code can be obtained by concatenating the \codepar{7,1,3} code with itself. The data qubits of the \codepar{49,1,9} code can be divided into 7 blocks, in which each block has 7 qubits that behave like the \codepar{7,1,3} code. The \codepar{49,1,9} code can be described by two types of stabilizer generators: the 1st-level generators of the form $g^x_i \otimes \tilde{I}^{\otimes 6}, g^z_i \otimes \tilde{I}^{\otimes 6}, \tilde{I} \otimes g^x_i \otimes \tilde{I}^{\otimes 5}, \tilde{I} \otimes g^z_i \otimes \tilde{I}^{\otimes 5}, \dots, \tilde{I}^{\otimes 6} \otimes g^x_i, \tilde{I}^{\otimes 6} \otimes g^z_i$ where $\tilde{I}=I^{\otimes 7}$ (which are the generators of the \codepar{7,1,3} code in each block), and the 2nd-level generators of the form

\begingroup
\setlength\arraycolsep{1pt}	
\begin{equation}
	\begin{matrix}
		\tilde{g}^x_1: &\tilde{X} &\tilde{I} &\tilde{I} &\tilde{X} &\tilde{I} &\tilde{X} &\tilde{X}, & \quad & \tilde{g}^z_1: &\tilde{Z} &\tilde{I} &\tilde{I} &\tilde{Z} &\tilde{I} &\tilde{Z} &\tilde{Z},\\
		\tilde{g}^x_2: &\tilde{I} &\tilde{X} &\tilde{I} &\tilde{X} &\tilde{X} &\tilde{I} &\tilde{X}, & \quad & \tilde{g}^z_2: &\tilde{I} &\tilde{Z} &\tilde{I} &\tilde{Z} &\tilde{Z} &\tilde{I} &\tilde{Z},\\
		\tilde{g}^x_3: &\tilde{I} &\tilde{I} &\tilde{X} &\tilde{I} &\tilde{X} &\tilde{X} &\tilde{X}, & \quad & \tilde{g}^z_3: &\tilde{I} &\tilde{I} &\tilde{Z} &\tilde{I} &\tilde{Z} &\tilde{Z} &\tilde{Z},
	\end{matrix}
\end{equation}%
\endgroup

where $\tilde{X}$ and $\tilde{Z}$ are logical $X$ and logical $Z$ operators of the \codepar{7,1,3} code. Note that the minimum weight of a 2nd-level generator is 12. Logical $X$ and logical $Z$ operators of the \codepar{49,1,9} code are of the form $\bar{X}=X^{\otimes 49}M$ and $\bar{Z}=Z^{\otimes 49}N$, where $M$ and $N$ are some stabilizers of the \codepar{49,1,9} code. Similar to the \codepar{7,1,3} code, it is possible to arrange the data qubits of the \codepar{49,1,9} code on a plane as shown in \cref{fig: stl2}\hyperlink{target:stl2}{b}. Logical Hadamard ($H$), phase ($S$), and CNOT gates can be implemented for the \codepar{49,1,9} code by applying the corresponding bit-wise operation on all qubits in the code block or on all pairs of qubits between two code blocks in the case of CNOT gate, thus any Clifford operation can be performed transversally.

\section{Noise models and fault-tolerant error correction decoders}
\label{sec:noise_model}

%\HG{Noise model: code cap, phenomenological, circuit level}

For any stabilizer code, one can perform error correction (EC) by first measuring the eigenvalues of the stabilizer generators. An $r$-bit string of the measurement results (where bits 0 and 1 correspond to $+1$ and $-1$ eigenvectors of each generator) is called an error syndrome. A mapping from a sequence of syndromes to a recovery operator is called an EC decoder. Since many possible combinations of faults can give rise to the same sequence of syndromes, the EC decoder succeeds if the actual data error due to faults and the recovery operator are the same up to a multiplication of some stabilizer, and the EC decoder fails if the actual data error due to faults and the recovery operator is off by a multiplication of some nontrivial logical operator.

In this work, we are interested in error correction in the circuit-level noise model. Here we assume that single-qubit and two-qubit gates can be faulty, where the gate faults are modeled by a single-qubit Pauli error $P \in \{X,Y,Z\}$ after each single-qubit gate with error probability $p/3$ each, and a two-qubit Pauli error $P_1\otimes P_2 \in \{I,X,Y,Z\}^{\otimes 2} \setminus \{I \otimes I\}$ after each two-qubit gate with error probability $p/15$ each. We also assume that a single-qubit preparation and a single-qubit measurement can be faulty, where the faults are modeled by a single-qubit bit-flip channel with error probability $p$ after each single-qubit preparation or before each single-qubit measurement. This noise model reflects a platform in which the strength of gate errors and qubit preparation and measurement errors are relatively large compared to the strength of idling qubit errors.

%\HG{Circuit with bare ancilla. Generally does not preserve distance in the circuit-level noise model. Goal is to make things FT}

\begin{figure}[tbp]
    \begin{subfigure}[b]{0.45\textwidth}
	\begin{equation}
		\Qcircuit @C=1em @R=.7em {
		& \ctrl{4} & \qw & \qw & \qw & \qw & \\
		& \qw & \ctrl{3} & \qw & \qw & \qw & \\
		& \qw & \qw & \ctrl{2} & \qw & \qw & \\
		& \qw  & \qw & \qw & \ctrl{1} & \qw & \\
		\lstick{\ket{0}} & \targ & \targ & \targ & \targ & \meter
		} \nonumber
	\end{equation}
        \captionsetup{justification=centering}
	\caption{}
	\label{fig:bare}
    \end{subfigure}
    \begin{subfigure}[b]{0.45\textwidth}
            \begin{equation}
    		\Qcircuit @C=1em @R=.7em {
    		& \ctrl{4} & \qw & \qw & \qw & \qw & \qw & \qw & \\
    		& \qw & \qw & \ctrl{3} & \qw & \qw & \qw & \qw & \\
    		& \qw & \qw & \qw & \ctrl{2} &  \qw & \qw & \qw & \\
    		& \qw & \qw & \qw & \qw & \qw & \ctrl{1} & \qw & \\
    		\lstick{\ket{0}} & \targ & \targ & \targ & \targ & \targ & \targ & \meter \\
    		\lstick{\ket{0}} & \gate{H} & \ctrl{-1} & \qw & \qw & \ctrl{-1} & \gate{H} & \meter
    		} \nonumber
    	\end{equation}
            \captionsetup{justification=centering}
    	\caption{}
    	\label{fig:flag}
        \end{subfigure}
    \caption{(a) A syndrome extraction circuit with a bare ancilla for measuring a stabilizer generator of the form $ZZZZ$. (b) A flag circuit for measuring the same stabilizer generator.}
\end{figure}
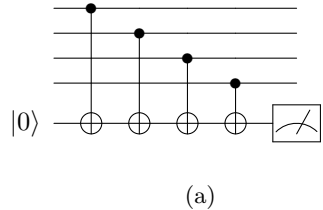
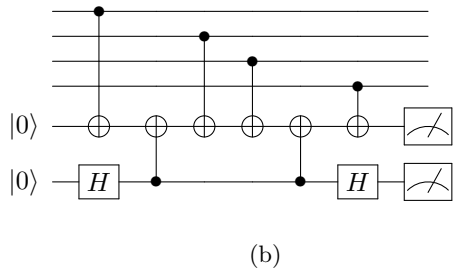

A naive way to measure the eigenvalue of a stabilizer generator is to use a syndrome extraction circuit with a single, so-called \textit{bare ancilla}, as shown in \cref{fig:bare}. However, in the circuit-level noise model, the number of faults that can be corrected by an error correction protocol with this kind of circuit might be less than the number of errors correctable by the code. This is because a single gate fault may cause a single-qubit error that can propagate throughout the protocol and become an error on the multiple data qubits. Here we want to ensure that the EC protocol is \emph{fault tolerant} according to the following definition from \cite{TL22}, which is extended from the definition of FTEC from \cite{AGP06} (see also \cite{PTHB23} for comparison between definitions from \cite{AGP06} and \cite{TL22}).

\begin{definition}{Fault-tolerant error correction \cite{TL22}}

	Let $t \leq \lfloor (d-1)/2\rfloor$ where $d$ is the distance of a stabilizer code. An error correction protocol is \emph{$t$-fault tolerant} if the following two conditions are satisfied:
	\begin{enumerate}
		\item For any input codeword with an error that can arise from $r$ faults before the protocol and corresponds to the trivial cumulative flag vector, if $s$ faults occur during the protocol with $r+s \leq t$, ideally decoding the output state gives the same codeword as ideally decoding the input state.
		\item If $s$ faults occur during the protocol with $s \leq t$, regardless of the number of faults that can cause the input error, the output state differs from any valid codeword by an error that can arise from $s$ faults and corresponds to the trivial cumulative flag vector.
	\end{enumerate}
	\label{def:FT_TL}
\end{definition}

In other words, whenever the total number of faults that occurred in the EC protocol is no more than the number of errors correctable by the underlying code, we want to ensure that the EC protocol can correct the input errors as expected and does not cause the output errors which are not correctable by the next EC cycle.

%\HG{Strategy: using flags. FT depends on CNOT ordering. Try to find good CNOT ordering. Verification tools from PTHB.}

To handle the error propagation issue, this work utilizes the flag FTEC scheme \cite{CR17a} in which a syndrome extraction circuit uses another flag ancilla to catch any fault that can lead to a high-weight error. An example of flag circuits is displayed in \cref{fig:flag}. It is possible to construct a distance-preserving flag FTEC scheme if, for a given set of syndrome extraction circuits, the fault set $\mathcal{F}_t$ with $t=\tau=\lfloor (d-1)/2 \rfloor$ is distinguishable; that is, any pair of fault combinations from up to $t$ faults lead to errors with the same syndrome and flag information (the \textit{full syndrome}) only when the errors are equivalent up to a multiplication of some stabilizer, or equivalently, none of the fault combinations from up to $d-1$ faults can lead to a non-trivial logical error on data qubits with trivial flag information \cite{TL22}. The fault-tolerant properties depend heavily on the structure of the syndrome extraction circuits, particularly the ordering of gates in the circuits. One goal of this work is to find a good gate ordering for the \codepar{49,1,9} concatenated Steane code, which gives a distinguishable fault set. How good a gate ordering can be found will be explained in \cref{sec:find_CNOT_ordering}.

%\HG{Suppose distinguishable, use space and time decoder as in PTHB. Lookup table, MIM, adaptive with flag, ZX separation.}

Given a good gate ordering that satisfies the distinguishability condition, we can construct an FTEC decoder using the ideas and techniques proposed in Ref. \cite{PTHB23}. Here, we consider an FTEC decoder consisting of two parts: the \emph{space decoder}, which maps a reliable syndrome from a single time slice to a recovery (Pauli) operator, and the \emph{time decoder}, which finds a reliable syndrome for the space decoder from a sequence of syndromes. The space decoder in this work is a lookup table decoder, which is a mapping between the full syndrome arising from each fault combination (from up to $t=4$ faults) and the corresponding $n$-qubit Pauli operator for recovery.
%an $n$-qubit Pauli operator arising from each fault combination (from up to $t=4$ faults) and its corresponding error syndrome and flag information. 
The lookup table from flag FTEC can be constructed from a given set of flag circuits by propagating single faults to the end of the circuit and generating all fault combinations of them exhaustively
%using the tools provided in 
\cite{PTHB23}. A lookup table decoder is fast and distance-preserving but requires a lot of memory to store the entries. However, this is not an issue in our case since the lookup table for the \codepar{49,1,9} code is 
%relatively small 
managable in size (34,404,345 items with 0.46 GB memory used) and can be constructed in our simulation.

When 5 or more faults occur, the measured full syndrome may not match the 
%corresponding 
full syndrome of any of the fault combinations 
%Pauli operator 
in the lookup table. To find a recovery operator in this case, we use the Meet-in-the-Middle (MIM) technique \cite{PTHB23} which performs a search at runtime during decoding, starting from any full syndrome that is missing from the table. Although the correction of 5 or more faults is not guaranteed, the MIM technique can significantly reduce the logical error rate and lead to a higher pseudothreshold.
%; see also \cite{PTHB23} for the full details of the MIM technique and its performance of the family of 6.6.6 color codes. 
In this work, we use a lookup table decoder and MIM with search radius 3 as a space decoder for the \codepar{49,1,9} code in the circuit-level noise model.

A syndrome obtained from a single round of full syndrome measurements may or may not correspond to the actual error on the data qubits, so repeated syndrome measurements are required. In this work, we use the ZX separated time decoder \cite{PTHB23}, which we briefly summarize here. The ZX time decoder is a two-tailed adaptive time decoder, an extension of the adaptive strong decoder from \cite{TPB23} for flag FTEC. The condition to stop repeated syndrome measurements for this time decoder changes dynamically depending on the previous measurement outcomes. In particular, this time decoder estimates the number of occurred faults from the full syndrome histories and deducts it from the targeted number of faults ($t$). The repetition stops when there exists a syndrome that is repeated more than the targeted number of faults. It has been shown \cite{TPB23} that with the two-tailed adaptive time decoder, the average number of syndrome measurement rounds for each QEC cycle is no more than $d$. In addition, the ZX time decoder leverages separated $X$ and $Z$ fault counting, in which $Z$-type generator measurements are performed before $X$-type generator measurements, to further improve the pseudothreshold.

\section{Finding the gate ordering for the concatenated Steane code}
\label{sec:find_CNOT_ordering}

In this section, we first describe a good gate ordering for the \codepar{49,1,9} concatenated Steane code, which gives a distinguishable fault set $\mathcal{F}_4$, then we explain how it can be found. Since the \codepar{49,1,9} concatenated Steane code is a CSS code in which $X$-type and $Z$-type generators can be chosen to be of the same form, for simplicity, we will only describe syndrome extraction circuits for measuring $Z$-type generators. In this work, a 1st-level $Z$-type generator of weight 4 and a 2nd-level $Z$-type generator of weight 12 are measured using a flag circuit similar to the one displayed in \cref{subfig:flag_w}. The circuits for measuring $X$-type generators are similar except that each CNOT gate that couples the data qubit and the syndrome ancilla is replaced by the gate in \cref{subfig:XNOT}.

\begin{figure}[tbp]
\begin{subfigure}[b]{0.45\textwidth}
\begin{equation}
\ \ \ \ \ \ \Qcircuit @C=1em @R=.7em {
                 & \ctrl{5} & \qw       & \qw      & \qw             & \qw       & \qw       & \qw       & \qw    \\                 
                 & \qw      & \qw       & \ctrl{4} & \qw             & \qw       & \qw       & \qw       & \qw    \\                
                 &          &           &          & \cdots          &           &           &           & \\
                 & \qw      & \qw       & \qw      & \qw             & \ctrl{2}  & \qw       & \qw       & \qw    \\                 
                 & \qw      & \qw       & \qw      & \qw             & \qw       & \qw       & \ctrl{1}  & \qw    \\
\lstick{\ket{0}} & \targ    & \targ     & \targ    & \qw             & \targ     & \targ     & \targ     & \meter \\
\lstick{\ket{0}} & \gate{H} & \ctrl{-1} & \qw      & \qw              & \qw       & \ctrl{-1} & \gate{H}  & \meter
 \gategroup{1}{4}{7}{6}{.7em}{_\}} \\ 
 &          &           &          & \text{$w-2$ gates}        &           &           &           & 
} \nonumber
\end{equation}
\captionsetup{justification=centering}
\caption{}
\label{subfig:flag_w}
\end{subfigure}
\begin{subfigure}[b]{0.45\textwidth}
\begin{equation}
\Qcircuit @C=1em @R=.7em {
& \gate{H} & \ctrl{1} & \gate{H} & \qw \\
& \qw & \targ & \qw &  \qw 
} \nonumber
\end{equation}
\captionsetup{justification=centering}
\caption{}
\label{subfig:XNOT}
\end{subfigure}
\caption{(a) A flag circuit for measuring a $Z$-type stabilizer generator of weight $w$ used in this work. A flag circuit for measuring an $X$-type stabilizer generator of weight $w$ can be obtained by replacing each CNOT gate that connects the data qubit to the syndrome ancilla with the gate in (b).}
\label{fig:flag_w}
\end{figure}
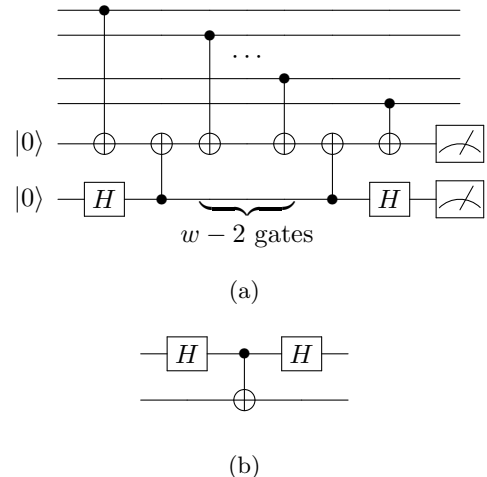
 
The ordering of CNOT gates in the flag circuits that give a distinguishable fault set $\mathcal{F}_4$ are determined by the diagram in \cref{fig:stl2 ordering} for the weight-12 operators. The ordering of the weight-4 stabilizer generators does not affect the fault distinguishability. The explicit ordering for all stabilizer generators used in this work is given in \cref{tab:stl2 orderings} in \cref{app:orderings}, with the qubit labeling given in \cref{fig:stl2 qubits}.

\begin{figure}[tbp]
\centering
\includegraphics[width=0.45\textwidth]{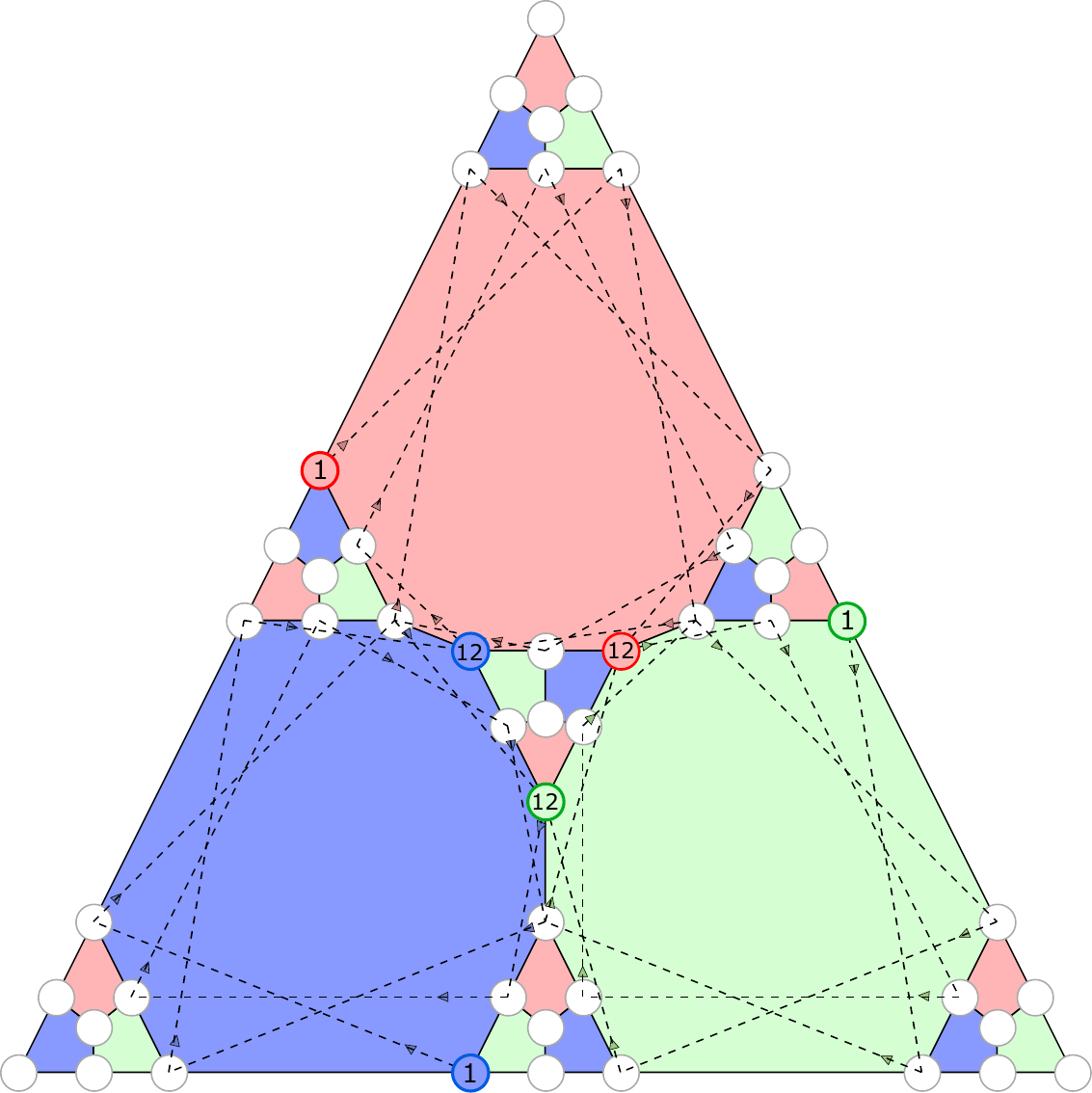}
\caption{A CNOT ordering for the \codepar{49,1,9} code that ensures that the fault set $\mathcal{F}_t$ is distinguishable. Here, we only display the ordering of the weight-12 stabilizer generators. The starting (1) and ending (12) qubits for each generator are highlighted with the same color.}
\label{fig:stl2 ordering}
\end{figure}

Our development of flag circuits is inspired by the weight-parity error correction (WPEC) technique and the flag circuits proposed in \cite{TL21} with an extension in \cite{TL22}. The main difference is that the circuits in the original proposal only allow a fault set $\mathcal{F}_3$ to be distinguishable, while our circuits allow a fault set $\mathcal{F}_4$ to be distinguishable. In other words, when applying to the \codepar{49,1,9} concatenated Steane code, a flag FTEC protocol that uses the flag circuits from the previous work can tolerate up to only 3 faults, while a protocol that uses the circuits in this work can tolerate up to 4 faults. The full details of our circuit development are described below.

%The weight parity error correction (WPEC) technique was originally proposed in \cite{TL21}
The key idea of WPEC is the following lemma:

\begin{figure*}[tbp]
	\centering
	\begin{subfigure}[b]{0.47\textwidth}
		\includegraphics[width=\textwidth]{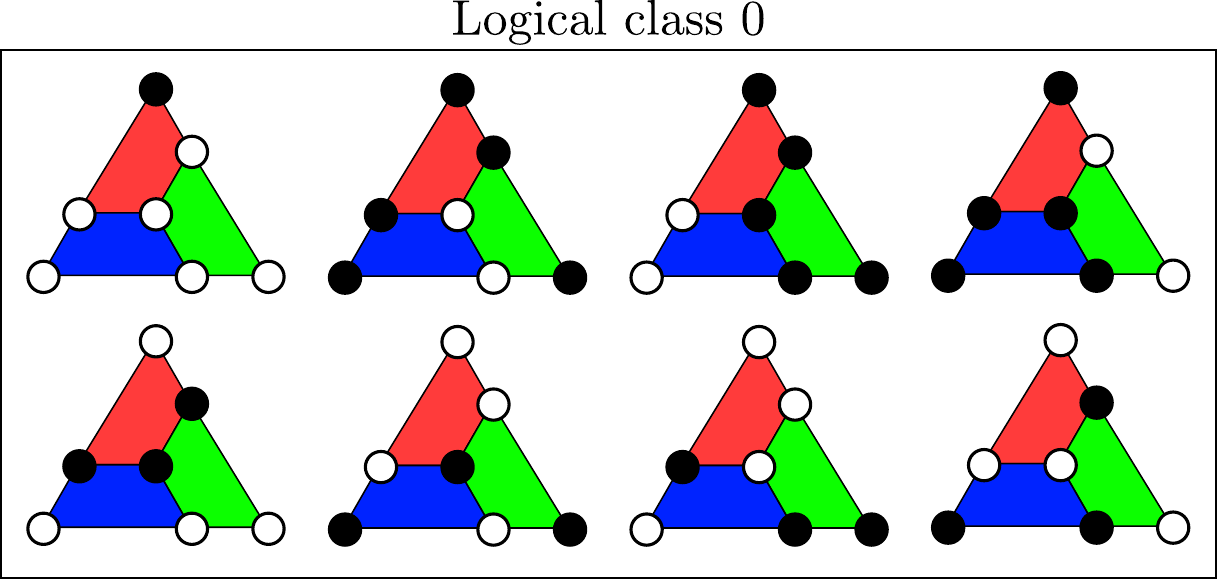}
		\captionsetup{justification=centering}
		%\caption{}
		\label{subfig:error_class0}
	\end{subfigure}	
	\begin{subfigure}[b]{0.47\textwidth}
		\includegraphics[width=\textwidth]{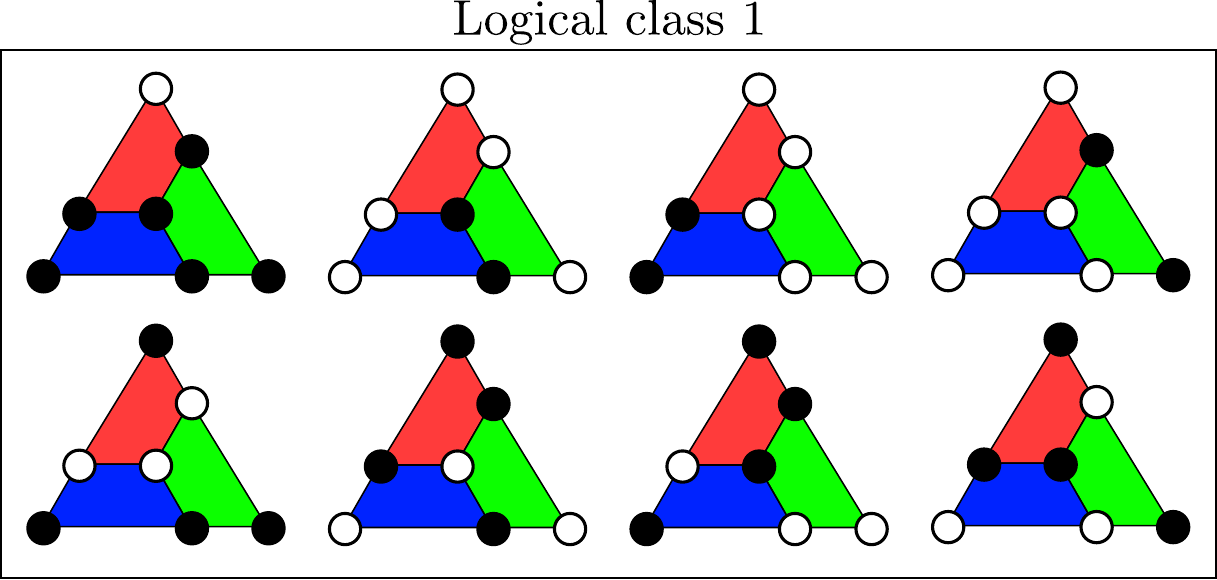}
		\captionsetup{justification=centering}
		%\caption{}
		\label{subfig:error_class1}
	\end{subfigure}
	\caption{All possible $X$-type errors on the \codepar{7,1,3} Steane code with the syndrome $\vec{s}=(100)$ (where the order of the syndrome bits corresponds to red (top), green (bottom-right), and blue (bottom-left) $Z$-type generators). Black and white vertices correspond to qubits with and without errors, respectively. When evaluating the weight parities on the same side of the triangle, errors with the same logical class always have the same weight parity.}
	\label{fig:error_class}%
\end{figure*}

\begin{lemma}
	\cite{TL22} Let $C$ be an \codepar{n,1,d} CSS code in which $n$ is odd and all stabilizer generators have even weight, and let $\mathcal{S}_x,\mathcal{S}_z$ be subgroups generated by $X$-type and $Z$-type generators of $C$, respectively. Also, let $X^{\otimes n}$ and $Z^{\otimes n}$ be logical $X$ and logical $Z$ operators. Suppose $E_1,E_2$ are Pauli errors of any weights with the same syndrome. 
	\begin{enumerate}
		\item In case that $E_1,E_2$ are $Z$-type errors, $E_1E_2 = M$ for some $M \in \mathcal{S}_z$ if and only if $E_1,E_2$ have the same weight parity.
		\item In case that $E_1,E_2$ are $X$-type errors, $E_1E_2 = M$ for some $M \in \mathcal{S}_x$ if and only if $E_1,E_2$ have the same weight parity.
	\end{enumerate}
	\label{lem:WPEC_TL22}%
\end{lemma}

%That is, for each error, its \emph{logical class} \HR{balint: this is a bit confusing, because logical class is defined to be relative to a set of CROs in the opt tools paper, but we don't mention CROs - one can figure out that CROs are chosen to be the odd weight recoveries, but maybe it's worth making it explicit?.} can be determined by its weight parity (more precisely, the logical class is 0 if the error has odd weight, and the logical class is 1 otherwise). By knowing the syndrome and the logical class of an unknown data error, we know exactly what recovery operation should be applied. 

That is, for errors with the same syndrome, we can determine the \emph{logical class} of each error by its weight parity (in this work, we let the logical class of each error be 0 if the error has odd weight, and let the logical class be 1 otherwise). By knowing the syndrome and the logical class of an unknown data error, we know exactly what recovery operation should be applied.

Consider error correction on the \codepar{49,1,9} concatenated Steane code, in which each block of 7 qubits behaves like the \codepar{7,1,3} Steane code. If the weight parity and the syndrome of an error in each block can be measured accurately, then an error in each block can be corrected regardless of the weight of the error. However, measuring the error weight parities of all blocks is not straightforward. In \cite{TL21}, the 2nd-level generators of the \codepar{49,1,9} code are chosen to be of the form $X^{\otimes 7}$ or $Z^{\otimes 7}$ on 4 blocks (the weight of each 2nd-level generator is 28). The weight parities of errors on all blocks are determined using the lookup table of the 1st and the 2nd level syndromes, and possible strings of weight parities. With the flag circuits for 1st-level generators and the non-flag circuits for 2nd-level generators given in \cite{TL21}, the fault set $\mathcal{F}_3$ is distinguishable. 
%(see also Section 5.2 of \cite{Tansuwannont21} for the connection between weight parities and fault set distinguishability).

\begin{figure*}[tbp]
	\centering
	\begin{subfigure}[b]{0.25\textwidth}
		\includegraphics[width=\textwidth]{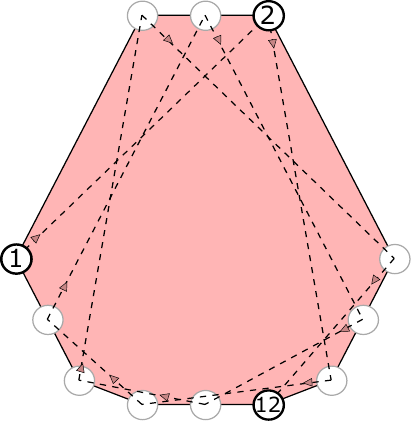}
		\captionsetup{justification=centering}
		\caption{}
		\label{subfig:stl2 good}
	\end{subfigure}
        %\par\medskip
        \hspace{6mm}
	\begin{subfigure}[b]{0.25\textwidth}
		\includegraphics[width=\textwidth]{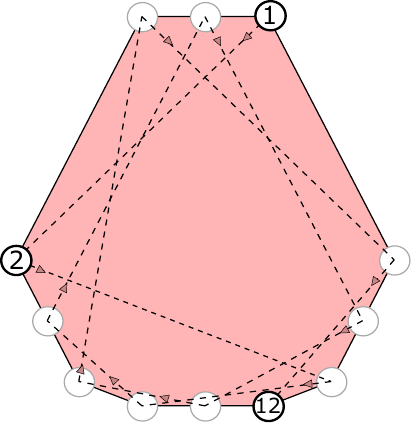}
		\captionsetup{justification=centering}
		\caption{}
		\label{subfig:stl2 bad}
	\end{subfigure}
\caption{A good CNOT ordering with effective distance 9 (a) and a bad CNOT ordering with effective distance 8 (b) for weight-12 generators of the \codepar{49,1,9} concatenated Steane code. The ordering starts with qubit 1 and 2, ends with qubit 12. The only difference between the two orderings is the swap of the first and second qubits.}
\end{figure*}

Note that the weight parity of an $X$-type (or a $Z$-type) error on each block of 7 qubits is related to its commutation and anticommutation relationship with the operator $Z^{\otimes 7}$ (or $X^{\otimes 7}$) on the same block. Also, $Z^{\otimes 7}$ and $X^{\otimes 7}$ are not logical $Z$ and logical $X$ operators of the minimum weight for the \codepar{7,1,3} Steane code. These facts suggest an alternative way to determine a logical class of each error and lead to a generalization of \cref{lem:WPEC_TL22} as follows:

\begin{lemma}
	Let $C$ be an \codepar{n,1,d} CSS code and let $\mathcal{S}_x,\mathcal{S}_z$ be subgroups generated by $X$-type and $Z$-type generators of $C$, respectively. Also, let $L_x$ be any logical $X$ operator and $L_z$ be any logical $Z$ operator. Suppose $E_1,E_2$ are Pauli errors of any weights with the same syndrome. 
	\begin{enumerate}
		\item In case that $E_1,E_2$ are $Z$-type errors, $E_1E_2 = M$ for some $M \in \mathcal{S}_z$ if and only if $[E_1,L_x]=[E_2,L_x]=0$ or $\{E_1,L_x\}=\{E_2,L_x\}=0$.
		\item In case that $E_1,E_2$ are $X$-type errors, $E_1E_2 = M$ for some $M \in \mathcal{S}_x$ if and only if $[E_1,L_z]=[E_2,L_z]=0$ or $\{E_1,L_z\}=\{E_2,L_z\}=0$.
	\end{enumerate}
	\label{lem:WPEC_new}%
\end{lemma}
\begin{proof}
	Here, we focus only on the first case where $E_1$ and $E_2$ are $Z$-type errors as the proof of the second case is similar. Because $E_1$ and $E_2$ are $Z$-type errors with the same syndrome, $E_1E_2$ is either a $Z$-type stabilizer or a logical $Z$ operator. Also, observe that $[L,E_1E_2]=LE_1E_2-E_1E_2L=\left((-1)^{b_1}-(-1)^{b_2}\right)E_1LE_2$ where $b_x=0$ if $E_x$ and $L$ commute and $b_x=1$ if they anticommute. Similarly, we have $\{L,E_1E_2\}=LE_1E_2+E_1E_2L=\left((-1)^{b_1}+(-1)^{b_2}\right)E_1LE_2$.
 
	$(\Leftarrow)$ In case that $[E_1,L_x]=[E_2,L_x]=0$ or $\{E_1,L_x\}=\{E_2,L_x\}=0$, we have $[L,E_1E_2]=0$. Thus, $E_1E_2$ must be a $Z$-type stabilizer in $S_z$.

    $(\Rightarrow)$ Since a pair of Pauli operators either commute or anticommute, the negation of ``$[E_1,L_x]=[E_2,L_x]=0$ or $\{E_1,L_x\}=\{E_2,L_x\}=0$" is ``$[E_1,L_x]=\{E_2,L_x\}=0$ or $\{E_1,L_x\}=[E_2,L_x]=0$". In either case, we have $\{L_x,E_1E_2\}=0$. Therefore, $E_1E_2$ must be a logical $Z$ operator. This completes the proof.
    
	%$(\Leftarrow)$ In case that $[E_1,L_x]=[E_2,L_x]=0$, we have $E_1L_x = L_xE_1$ and $E_2L_x = L_xE_2$. In case that $\{E_1,L_x\}=\{E_2,L_x\}=0$, we have $E_1L_x = -L_xE_1$ and $E_2L_x = -L_xE_2$. Either way, the pair of equations give $(E_1L_x)(L_xE_2) = (L_xE_1)(E_2L_x)$. That is, $E_1E_2 = L_x(E_1E_2)L_x$, which is equivalent to $[E_1E_2,L_x]=0$. Therefore, $E_1E_2$ must be a $Z$-type stabilizer in $S_z$.
	
	%$(\Rightarrow)$ Since a pair of Pauli operators either commute or anticommute, the negation of ``$[E_1,L_x]=[E_2,L_x]=0$ or $\{E_1,L_x\}=\{E_2,L_x\}=0$" is ``$[E_1,L_x]=\{E_2,L_x\}=0$ or $\{E_1,L_x\}=[E_2,L_x]=0$". In case that $[E_1,L_x]=\{E_2,L_x\}=0$, we have $E_1L_x = L_xE_1$ and $E_2L_x = -L_xE_2$. In case that $\{E_1,L_x\}=[E_2,L_x]=0$, we have $E_1L_x = -L_xE_1$ and $E_2L_x = L_xE_2$. Either way, we find that $(E_1L_x)(L_xE_2) = -(L_xE_1)(E_2L_x)$, which is equivalent to $\{E_1E_2,L_x\}=0$. Thus, $E_1E_2$ must be a logical $Z$ operator. This completes the proof.
\end{proof}

Note that in contrast to \cref{lem:WPEC_TL22}, \cref{lem:WPEC_new} does not require $n$ to be odd or all stabilizers to have even weight.

Consider possible $X$-type errors on the \codepar{7,1,3} Steane code with the same syndrome $\vec{s}=(100)$ depicted in \cref{fig:error_class} as an example (where the order of the syndrome bits corresponds to red (top), green (bottom-right), and blue (bottom-left) $Z$-type generators). Here, we can see that two errors are logically equivalent if and only if their weight parities \emph{evaluated on the same side of the triangle} are equal. This comes from the fact that the $Z$-type operator of weight 3 on each side of the triangle is a logical $Z$ operator of the \codepar{7,1,3} Steane code. That is, if we know exactly on which side of the triangle the weight parity of the error is being measured, we can find the logical class of the error accurately.

\cref{lem:WPEC_new} and the example above suggest that we can choose the 2nd-level generators of the \codepar{49,1,9} concatenated Steane code to be operators of weight 12 (as described in \cref{sec:QECC}) instead of weight 28 as in \cite{TL21}. If only a single fault occurs, a recovery operator for each block of 7 qubits can be found using the 1st-level syndrome from generators in each block, together with the 2nd-level syndrome, which provides the logical class information (since we know exactly which side of each triangle is being measured). %Our goal is to find a recovery operator when up to 4 faults occur. However, finding the logical class information for each block could be more complicated when the number of faults increases. Thus, we will instead solve an equivalent problem: finding a CNOT ordering for flag circuits that give a distinguishable fault set. 
Our goal is to find a recovery operator when up to 4 faults occur. Finding the logical class information for each block could be more complicated when the number of faults increases since each 2nd-level syndrome bit provides the ``sum" of the logical classes from all blocks that the 2nd-level generator touches. One way to find the recovery operators is to iterate through all possible fault combinations arising from up to 4 faults and build a mapping between each full syndrome and the corresponding logical classes from all blocks (similar to the WPEC technique in \cite{TL21}), where the possible values depend on the CNOT ordering. Note that this is equivalent to finding a CNOT ordering for flag circuits that give a distinguishable fault set, a problem in which the tools provided in \cite{PTHB23} can solve well if the code size is not too large.
%\HR{This is where I have an issue with the whole thing. We spent a whole page understanding WPEC, and then suddenly we are looking at CNOT ordering..., which are only saying "spread them evenly across the subblocks". What is the ordering in [1]? How does it inspire this one?} 

\begin{figure*}[tbp]
\begin{subfigure}[b]{0.49\textwidth}
\includegraphics[width=\textwidth]{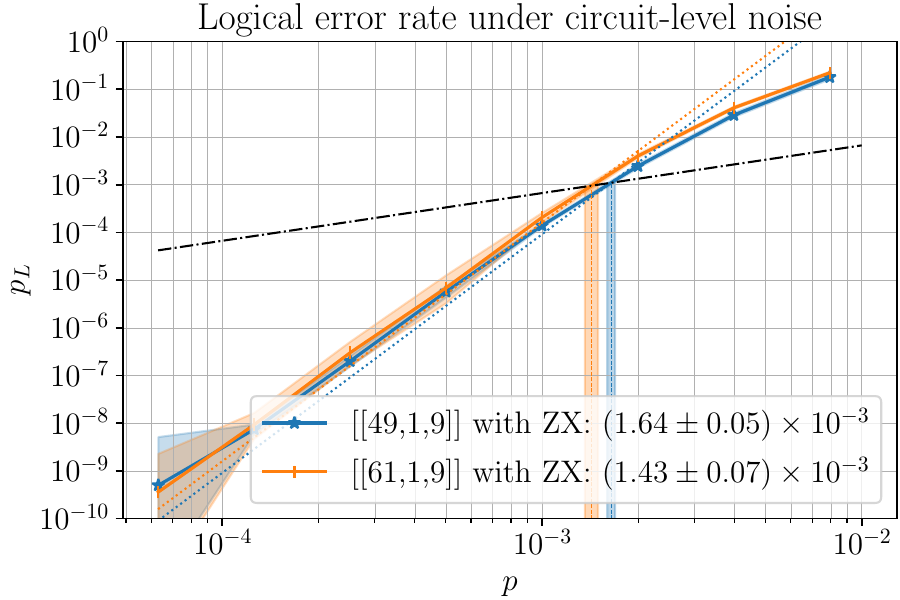}
\captionsetup{justification=centering}
\caption{}
\label{fig:circuit_level}
\end{subfigure}
\hspace{1mm}
\begin{subfigure}[b]{0.49\textwidth}
\includegraphics[width=\textwidth]{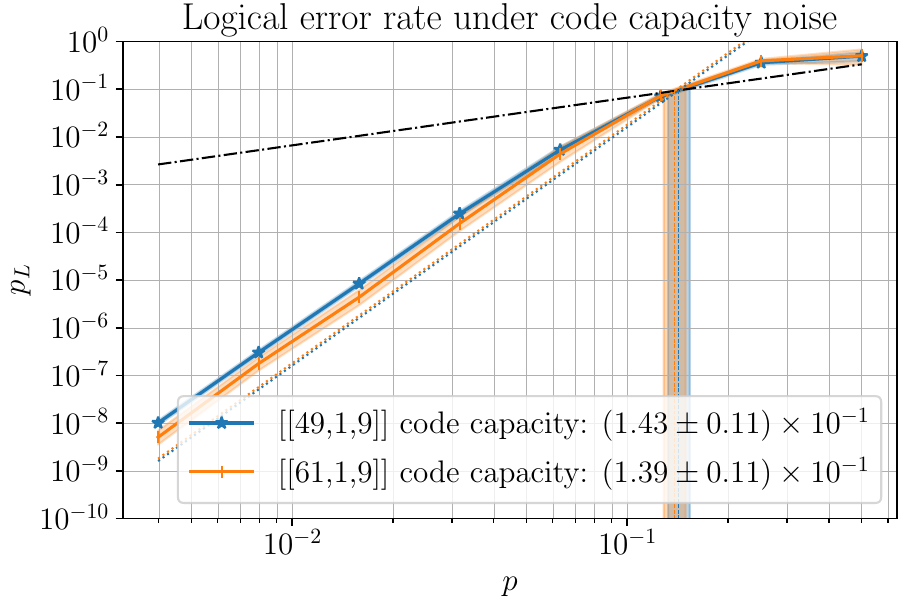}
\captionsetup{justification=centering}
\caption{}
\label{fig:cc}
\end{subfigure}
\vspace{1mm}
\caption{(a) The circuit level performance of the \codepar{49,1,9} concatenated Steane code and the \codepar{61,1,9} 6.6.6 color code. The two codes are very close in performance, however, the pseudothreshold of the concatenated Steane code is roughly 14\% higher than that of the 6.6.6 color code. The pseudothreshold is against the $p=2p/3$ line, denoted by the green dashed line on the graph. The dotted lines, which serve as guides to the eyes, have the same slope as $p^5$ and intersect at the pseudothresholds, showing that our protocol is distance-preserving. (b) Logical vs physical error rates under code capacity error model for the \codepar{49,1,9} concatenated Steane code and the \codepar{61,1,9} 6.6.6 color code. The pseudothreshold is against the $p=2p/3$ line, denoted by the black dash-dot line on the graph. The dotted lines have the same slope as $p^5$ and intersect at the pseudothresholds, showing that both error correcting codes have distance 9.} 
\end{figure*}

To find such CNOT ordering, we use an idea similar to the construction in \cite{TL21}: the CNOT ordering for each circuit is chosen in the way that possible errors arising from each single fault are distributed to as many blocks as possible. This is to avoid the case that many errors concentrate in a single block but the flag bits are trivial as much as possible. Using the tools provided in \cite{PTHB23}, we can verify that the ordering provided in \cref{tab:stl2 orderings} gives a distinguishable fault set $\mathcal{F}_4$. The ordering for any weight-12 generator is illustrated in \cref{subfig:stl2 good}. It should be noted that our circuits for measuring 2nd-level generators are flag circuits, in contrast to the circuits in \cite{TL21}, which are non-flag circuits.

We point out that choosing the gate ordering that distribute errors alone cannot guarantee the distinguishability of $\mathcal{F}_4$. We find that when using, for example, the ordering in \cref{subfig:stl2 bad} for all weight-12 generators, $\mathcal{F}_4$ is not distinguishable; the bad ordering only differs from the good ordering by a single swap of the roles of two qubits. This swap results in the decrease of the effective distance from 9 to 8, showcasing how sensitive the protocol is to ordering.

\section{Numerical results} \label{sec:numerics}

In our numerical simulations of the flag FTEC protocol on the \codepar{49,1,9} concatenated Steane code, we collect data to plot the logical error rate $p_L$ as a function of the physical error rate $p$ under the circuit-level noise model. The preparation of the logical zero states on the data qubits at the beginning of the protocol is assumed to be perfect. When applicable, gate noise is implemented as single and two-qubit depolarizing instructions and preparation and measurement noises are implemented as bit-flip noise on a single qubit after preparation and before measurement operations. The simulator uses Pauli frame propagation of noise terms through the Clifford-only operations of the syndrome extraction circuits. In each round of syndrome measurements, we use flag circuits with the gate ordering described in \cref{app:orderings}. After executing the necessary number of noisy rounds determined by the time decoder, the full syndrome obtained from the time decoder is passed to the space decoder, which then gives us a recovery operator. We then perform an ideal syndrome measurement, apply an ideal recovery operator, and calculate the true eigenvalue of the logical $Z$ observable on the output state. If the eigenvalue of the logical $Z$ observable is $-1$, we have a logical error. 

For each physical error rate $p$, we collect up to $5 \times 10^{10}$ sample points. We stop the collection early when the number of logical errors reaches 1000. Error bars are reported as shaded areas in our figures, and they should be interpreted as the most likely area for the true value of the logical error rate. The reported $p_L$ value is the sample mean. Using the binomial likelihood function we can estimate the probability that a given value $p_L$ is the correct probability of logical error given the number of samples and number of logical errors. The upper and lower bounds of the shaded area are defined as $p_L$ values that are $10^3$ times less likely than the sample mean, the maximum likelihood estimator.

\begin{table*}
\centering
\begin{tabular}{|c|c|c|c|c|}
\hline
\multirow{2}{*}{Error model} & \multirow{2}{*}{Code}  & \multicolumn{3}{c|}{Number of locations} \\
\cline{3-5}
 &  & data error & gate &  preparation \& measurement \\
\hline
\hline
\multirow{2}{*}{Circuit-level} & \codepar{61,1,9} & 0 & 276 & 120 \\
\cline{2-5}
 & \codepar{49,1,9} & 0 & 216 & 96 \\   
\hline 
\multirow{2}{*}{Code capacity} & \codepar{61,1,9} & 61 & 0 & 0 \\
\cline{2-5}
 & \codepar{49,1,9} & 49 & 0 & 0 \\  
\hline
\end{tabular}
\caption{Number of fault locations under the circuit-level and code capacity error models for the \codepar{49,1,9} concatenated Steane code and the \codepar{61,1,9} 6.6.6 color code. Data error refers to single-qubit depolarizing noise terms on each data qubit before a round of error correction. Gate faults correspond to single-qubit and two-qubit depolarizing noise terms after single-qubit and two-qubit gates, respectively. Preparation and measurement faults are bit-flip noise terms after ancilla reset and before ancilla measurement operations, both in the Z-basis. The number of locations for the circuit-level noise model is per round of error correction and per X or Z-type of stabilizer measurement circuits, meaning for a full syndrome measurement round, one needs to multiply the numbers by two.}
\label{tab:locations}
\end{table*}

We generate the noisy circuit definitions in Python using Cirq \cite{Cirq22}, sample the circuits using Stim \cite{Gidney21}, and decode them with our C++ lookup table decoder \cite{PTHB23}. Pseudothresholds are calculated based on the intersection with the $p_L=2p/3$ line of the linearly interpolated sample mean. For pseudothreshold errors, we calculate the intersections with the $p_L=2p/3$ line of the upper and lower bounds and use the one that is further from the pseudothreshold. We note that the $p_L=2p/3$ line represents the infidelity of any single-qubit pure state when it is sent through the depolarizing channel in which each single-qubit Pauli error occurs with probability $p/3$.

For comparison, we also perform numerical simulations of the flag FTEC protocol on the \codepar{61,1,9} 6.6.6 color code of distance 9. The 6.6.6 color code of distance $d$ is an \codepar{(3d^2+1)/4,1,d} self-orthogonal CSS code, whose qubits can also be laid out on a plane \cite{BM06}. When performing flag FTEC using flag circuits similar to the one in \cref{fig:flag_w}, it has been proved that the code distance is preserved regardless of the gate ordering being used \cite{CKYZ20,TL22}. The exact gate ordering for the \codepar{61,1,9} code that we use is described in \cref{tab:tccd9 orderings} in \cref{app:orderings}. In this work, we used the simulation data from \cite{PTHB23} and collected some more data points using the same simulator to reduce the error bars.

The plots of $p_L$ versus $p$ under the circuit-level noise model for the \codepar{49,1,9} concatenate Steane code and the \codepar{61,1,9} 6.6.6 color code are shown in \cref{fig:circuit_level}. Using the lookup table decoder for flag FTEC, the MIM technique, the two-tailed adaptive time decoder, and the ZX separated counting strategy (as previously described in \cref{sec:noise_model}), the \codepar{49,1,9} code achieves a pseudothreshold of $(1.64 \pm 0.05) \times 10^{-3}$, which is slightly better than a pseudothreshold of $(1.43 \pm 0.07) \times 10^{-3}$ for the \codepar{61,1,9} code. The separation in logical error rates between the two codes disappears when the physical error rate is below $p=1.0 \times 10^{-3}$. 

While the difference is small, this result is surprising because in the code capacity error model reported in \cite{SAB22}, the ranking is the opposite, as the \codepar{61,1,9} code slightly outperforms the \codepar{49,1,9} code. In order to exclude the possibility of the difference between our lookup table-based decoder and the trellis decoder used in that work, we also simulate the logical error rates for the two codes of interest under the code capacity noise model (where the gate, preparation, and measurement faults are absent). In this case, we use a lookup table decoder for qubit errors only as a space decoder, and time decoding is not necessary since the syndrome measurements can be assumed to be perfect in this noise model. Our results are reported in \cref{fig:cc}. While the pseudothresholds of the two codes are in the same range, the separation in logical error rates is clear when the physical error rate is below $p=1.0 \times 10^{-2}$.  This confirms that the \codepar{61,1,9} code outperforms the \codepar{49,1,9} code in the code capacity noise model.

% \HR{TODO: get the different average error rates per number of faults at different $p$ levels} 

\section{Discussion}\label{sec:discussion}

It is natural to ask how the \codepar{49,1,9} concatenated Steane code can beat the \codepar{61,1,9} 6.6.6 color code in the circuit-level noise model. We conjecture that the number of locations where faults can occur per round (noise instructions in Stim) plays an important role in our settings. \cref{tab:locations} displays the number of possible locations for different codes and different noise models. 

Due to the logic of the adaptive syndrome measurement techniques, when a fault occurs in a single round, it is very likely that the syndrome of that round is different from the syndrome of the previous round, causing the repeated syndrome measurements to continue. That is, more locations per round likely lead to more rounds and also more total locations in the whole protocol.

Note that the code distance is preserved in both codes and both noise models, so all fault combinations with up to 4 faults can be successfully corrected. Since the cases with 5 faults are more likely to happen than the cases with 6 or more faults, we conjecture that the logical error rate for each code and each model is mainly determined by the proportion of fault combinations with 5 faults that lead to decoding failure. The proportion varies with the noise model, so the ranking in the circuit-level noise model could be different from the ranking in the code capacity noise model. The reason that the \codepar{61,1,9} code performs worse in the circuit-level noise model could be because its protocol has more total locations on average compared to the  \codepar{49,1,9} code in the same noise model, which leads to more possible decoding failure cases.

Another thing to point out is that the \codepar{49,1,9} concatenated Steane code has smaller average weight per stabilizer generator compared to the \codepar{61,1,9} 6.6.6 color code; the \codepar{49,1,9} code has 6 generators of weight 12 and 42 generators of weight 4, while the \codepar{61,1,9} code has 36 generators of weight 6 and 24 generators of weight 4. This could be another factor that makes the \codepar{49,1,9} code perform better in the circuit-level noise model.  

It should be noted that the 4.8.8 color code of distance 9 \cite{BM06} is also a \codepar{49,1,9} self-orthogonal CSS code, so it is natural to ask whether it performs better than the \codepar{49,1,9} concatenated Steane code in the circuit-level noise model. However, as shown in \cref{app:488nogo}, it is impossible to construct a flag FTEC protocol with only a single flag qubit per generator that preserves the code distance for the 4.8.8 color codes. Since more flags are required when decoding in the circuit-level noise model, the 4.8.8 color code seems worse than the concatenated Steane code in terms of qubit efficiency. We still do not know whether the 4.8.8 color code could perform better than the concatenated Steane code. It is possible that the 4.8.8 color code would perform worse since more flag qubits involved would lead to more locations per round. However, from our simulations it is clear that the situation is subtle and knowing the number of locations alone might not be enough to predict the performance.

\iffalse
\begin{table*}
\centering
\begin{tabular}{|c|c|c|c|}
\hline
Error model & Code/protocol  & # locations per X/Z stabilizers & input noise \\
\hline
\hline
Circuit-level & Single flag \codepar{61,1,9} & 120 MR + 276 Gates= 396 & 0 \\  
Phenomenological & Bare ancilla \codepar{61,1,9} & M 30 & 61 \\  
Code capacity & Bare ancilla \codepar{61,1,9} & 0 & 61 \\  
\hline 
Circuit-level & Single flag \codepar{49,1,9} & 96 MR + 216 Gates= 312 & 0 \\  
Phenomenological & Bare ancilla \codepar{49,1,9} & M 24 & 49 \\  
Code capacity & Bare ancilla \codepar{49,1,9} & 0 & 49 \\  
\hline
\end{tabular}

\caption{Number of fault locations under different error models for the \codepar{41,1,9} and the \codepar{61,1,9} codes. MR stands for noise instructions before measurement and after reset. The input noise column shows the weight of the depolarizing noise before each round of error correction - the number of data qubits.}
\label{tab:locations}
\end{table*}
\fi

\section{Conclusion and outlook}\label{sec:conclusion}

We found a distance-preserving, flag FTEC protocol with only two ancilla qubits per generator for the \codepar{49,1,9} concatenated Steane code. If we put physical connectivity constraints aside, the concatenated Steane code outperforms its topological siblings, the \codepar{61,1,9} 6.6.6 color code and the \codepar{49,1,9} 4.8.8 color code, under circuit-level depolarizing noise without idling noise. The 6.6.6 color code requires at least 12 more data qubits than the concatenated Steane code and still has lower performance near the pseudothreshold. We also showed that the 4.8.8 code cannot have a distance-preserving flag FTEC protocol with only a single flag qubit per generator. 

To our knowledge, no self-orthogonal CSS code of distance 9 with less than 49 data qubits has been found yet. Thus, we believe that the level of error suppression that the  \codepar{49,1,9} concatenated Steane code can achieve might be a promising target to demonstrate on an experimental platform where gate errors dominate over idling noise and at least 51 qubits with all-to-all qubit connectivity are available (assuming that fast measurement and reset on ancilla qubits are possible). Early fault tolerant demonstrations of Rydberg atom systems \cite{Bluvstein_2024} and trapped ion systems \cite{Wang_2023,Huang_Brown_Cetina_2023,egan_fault-tolerant_2021,Postler_Butt_2023,Ryan-Anderson_Brown_2022} promise to have the scale as well as error characteristics to make them a promising target for our scheme. We would like to emphasize that the performances of the concatenated Steane code and the 6.6.6 color code are very close under the considered noise model. This means that exposed to realistic noise, the ranking of the two codes might change. This issue might be especially pertinent on a 2D array of neutral atoms, where the concatenated Steane code might need more qubit movements than the 6.6.6 code. The more movements are expected due to the higher connectivity of the weight-12 stabilizer generators, and these moves might result in a significant increase in idling noise.

Our protocol would need to be significantly modified to work on platforms with high idling noise and constrained connectivity, for example, superconducting qubit architectures. We point out that the planar layout of the code allows for embedding it in a 2D architecture. However, for the weight-12 stabilizer generators, the single-flag scheme imposes degree 13 connectivity on the ancilla qubits, which is much higher than the currently available devices \cite{Morvan_2023,Kim_2023}. It is possible to preserve planarity and low connectivity at the cost of an increased number of ancilla qubits. Dominant idling noise requires further optimizations to minimize the depth of our circuit and, hence, to reduce the number of locations where idling noise can occur. We leave this exploration to future work. 

Our methods handle the concatenated Steane code as a regular color code, in the sense that we do not concatenate the ancilla qubits and do not make explicit use of the concatenated structure in our decoder, as opposed to what a level-by-level hard decoder for a concatenated code in \cite{AGP06} would do. This approach can be considered as an example of a circuit-level soft decoder for a concatenated code. Poulin showed that in the code capacity noise model, soft decoding using belief propagation is efficient and optimal for concatenated codes \cite{Poulin_2006}. A further path of inquiry could be to find a circuit-level soft decoder that has better scaling than our lookup table-based method. 

The role of CNOT ordering has proved to be critical for preserving distance in the concatenated Steane code. CNOT ordering might also have an impact on the success probability of decoding fault combinations consisting of more than $\lfloor (d-1)/2 \rfloor$ faults for both the concatenated Steane code and the 6.6.6 color code. It is unclear whether there exist CNOT orderings that can close the gap between the two codes or even restore the ranking found under code capacity error models. The exact conditions of gate orderings that can give a distinguishable fault set is still an open question. Another interesting research direction would be investigating whether it is possible to generalize our techniques to a concatenated Steane code with more levels of concatenation. It is also intriguing to explore whether our techniques can be applied in the recently discovered high-performance concatenated architectures that concatenate different codes at different levels \cite{Yamasaki_Koashi_2024, Yoshida_Tamiya_Yamasaki_2024}. If possible, the combined extraction of multiple levels with only a few extra flag qubits could result in significant qubit savings.

Our work highlights the fact that the code performance ranking under the code capacity noise model does not predict performance ranking under the circuit-level noise model. Under the code capacity noise model, the \codepar{61,1,9} 6.6.6 color code performs better than the \codepar{49,1,9} concatenated Steane code, and the \codepar{49,1,9} 4.8.8 color code has an even higher threshold than the \codepar{61,1,9} 6.6.6 color code \cite{SAB22}. In circuit-level noise restricted to using a single flag qubit per generator, the concatenated Steane code comes out as the best performer. While we conjecture that the role of CNOT ordering, the number of fault locations, and the average weight per generator all play a role in creating these effects, more exploration and data are needed for a concrete theory.

\section{Data and code availability}

All source code to reproduce the data and the actual data in this work are available to download \cite{pato_2024_13312437}. 

\section{Acknowledgements}

We would like to thank Shilin Huang and Narayanan Rengaswamy for helpful discussions. We also thank Andrew Landahl for suggesting looking into the 4.8.8 color codes. The work was supported by the Office of the Director of National Intelligence - Intelligence Advanced Research Projects Activity through an Army Research Office contract (W911NF-16-1-0082), the Army Research Office Multidisciplinary University Research Initiative (MURI) program (W911NF-16-1-0349), the Army Research Office (W911NF-21-1-0005), and the National Science Foundation Institute for Robust Quantum Simulation (QLCI grant OMA-2120757).

\appendix

\section{Explicit CNOT orderings}\label{app:orderings}

For our numerical simulations, we used the indexing for the qubits in the \codepar{49,1,9} concatenated Steane code as displayed in \cref{fig:stl2 qubits}. The orderings of the data CNOT operators for each stabilizer generator are then defined in \cref{tab:stl2 orderings} using these indices. 

Similarly, we used the indexing for the qubits in the \codepar{61,1,9} 6.6.6 color code as displayed in \cref{fig:tccd9 qubits}. The orderings of the data CNOT operators for each stabilizer generator are then defined in \cref{tab:tccd9 orderings} using these indices. 

\begin{figure}[htbp]
\centering
\includegraphics[width=0.48\textwidth]{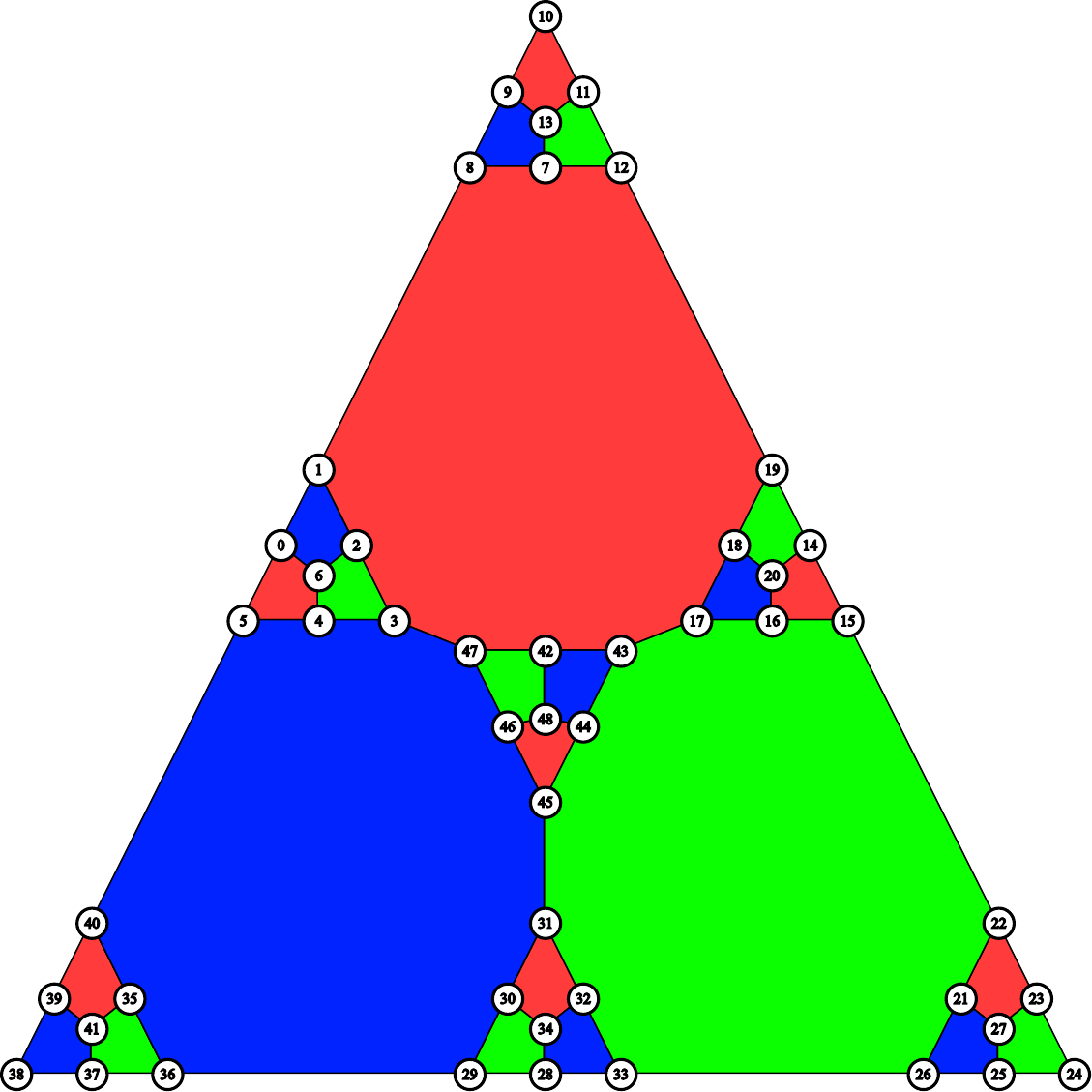}
\caption{The qubit indices of the \codepar{49,1,9} concatenated Steane code used in our numerical simulations. The CNOT orderings for the stabilizer generators are based on these numbers. }
\label{fig:stl2 qubits}
\end{figure}

\begin{figure}[htbp]
\centering
\includegraphics[width=0.48\textwidth]{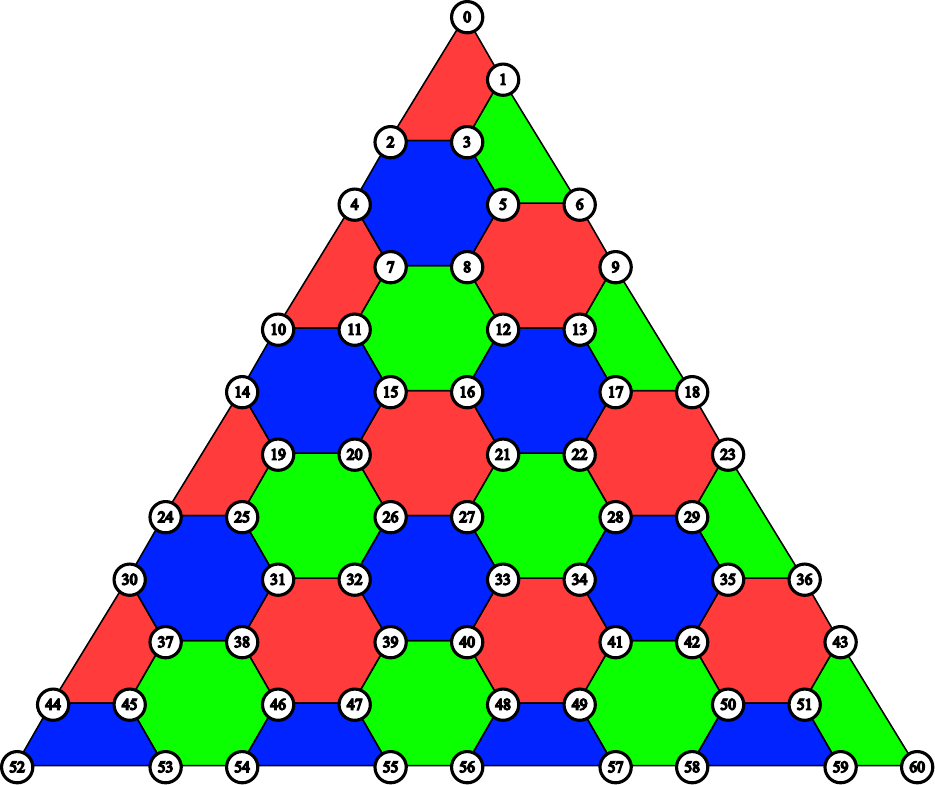}
\caption{The qubit indices of the \codepar{61,1,9} 6.6.6 color code used in our numerical simulations. The CNOT orderings for the stabilizer generators are based on these numbers. }
\label{fig:tccd9 qubits}
\end{figure}

\begin{table}[htbp]
\small
\begin{tabular}{|c|c|c|}  
    \hline
    Red & Green & Blue \\
    \hline 
[4,5,0,6] & [2,3,4,6] & [0,1,2,6] \\ [2pt]
\hline
[9,10,11,13] & [11,12,7,13] & [7,8,9,13] \\[2pt]
\hline
[14,15,16,20] & [18,19,14,20] & [16,17,18,20] \\[2pt]
\hline
[21,22,23,27] & [23,24,25,27] & [25,26,21,27] \\[2pt]
\hline
[30,31,32,34] & [28,29,30,34] & [32,33,28,34] \\[2pt]
\hline
[39,40,35,41] & [35,36,37,41] & [37,38,39,41] \\[2pt]
\hline
[44,45,46,48] & [46,47,42,48] & [42,43,44,48] \\[2pt]
\hline
$\begin{matrix}[1,\!12,\!17,\!47,\\2,\!7,\!18,\!42,\\3,\!8,\!19,\!43]\end{matrix}$ & $\begin{matrix}[15,\!26,\!31,\!43,\!\\16,\!21,\!32,\!44,\!\\17,\!22,\!33,\!45]\end{matrix}$ & $\begin{matrix}[29,\!40,\!3,\!45,\!\\30,\!35,\!4,\!46,\!\\31,\!36,\!5,\!47]\end{matrix}$ \\ [2pt]
\hline
\end{tabular}
\centering
\caption{CNOT ordering for each generator of the \codepar{49,1,9} concatenated Steane code.}
\label{tab:stl2 orderings}

\end{table}
 
\begin{table}[htbp]
\small
\begin{tabular}{|c|c|c|} 
    \hline  
    Red &Green & Blue \\
    \hline 
[0,2,3,1]&[1,6,5,3]&[44,45,53,52] \\ [2pt]
\hline 
[4,10,11,7] &[9,18,17,13]&[46,47,55,54]\\[2pt]
\hline 
[14,24,25,19]&[23,36,35,29]&[48,49,57,56]\\ [2pt]
\hline 
[30,44,45,37]&[43,60,59,51]&[50,51,59,58]\\    [2pt] 
\hline 
[5,6,9,13,12,8]&[7,8,12,16,15,11]&[2,3,5,8,7,4]\\  [2pt]   
\hline 
[15,16,21,27,26,20]& [19,20,26,32,31,25]&[10,11,15,20,19,14]\\ [2pt]
\hline 
[17,18,23,29,28,22]& [21,22,28,34,33,27]&[12,13,17,22,21,16]\\ [2pt]
\hline 
[31,32,39,47,46,38]& [37,38,46,54,53,45]&[24,25,31,38,37,30]\\ [2pt]
\hline 
[33,34,41,49,48,40]& [39,40,48,56,55,47]&[26,27,33,40,39,32]\\ [2pt]
\hline 
[35,36,43,51,50,42]& [41,42,50,58,57,49]&[28,29,35,42,41,34]\\ [2pt]
\hline
\end{tabular}
\centering
\caption{CNOT ordering for each generator of the \codepar{61,1,9} 6.6.6 color code.}
\label{tab:tccd9 orderings}

\end{table}

\section{No-go theorem for distance-preserving flag FTEC with single flag ancilla for the 4.8.8 color codes}
\label{app:488nogo}

In this section, we prove that for any 4.8.8 color code of distance $d \geq 5$, it is impossible to construct a flag FTEC protocol with a single flag ancilla that preserves the code distance. The proof has two main steps: (1) We perform exhaustive search on all possible CNOT orderings for the 4.8.8 color code of distance 5, given that each syndrome extraction circuit have only one flag ancilla (similar to the circuit in \cref{fig:flag_w}). In this step, we show that for any possible single-flag syndrome extraction circuit and any possible CNOT ordering, at least one logical operator of weight 5 with trivial flag bits can arise from 4 or fewer faults \footnote{A code for our proof is published in a Colab python notebook: \url{https://colab.research.google.com/drive/115qVvd8zvEF8JikAHIKLtnGuTrjsMUZR}}. (2) We extend the result to the code of distance $d=5+2i$ with $i=1, 2, 3, \dots$ by relating each CNOT ordering for that code with a CNOT ordering for the code of distance 5 of a similar form. In this step, we show that for any possible ordering for the code of distance $d=5+2i$, at least one logical operator of weight $5+2i$ with trivial flag bits can arise from $4+2i$ or fewer faults. 

%\footnote{A code for our proof is published in a Colab python notebook: \href{https://colab.research.google.com/drive/115qVvd8zvEF8JikAHIKLtnGuTrjsMUZR}{https://colab.research.google.com/drive/} \href{https://colab.research.google.com/drive/115qVvd8zvEF8JikAHIKLtnGuTrjsMUZR}{115qVvd8zvEF8JikAHIKLtnGuTrjsMUZR}}.

In Step (1), we start by observing that the 4.8.8 color code of distance 5 has a single weight-8 stabilizer generator, and all the other stabilizer generators have weight 4. For each weight-4 generator, we choose its flag circuit to be the circuit in \cref{fig:flag}, the best possible flag circuit for each weight-4 generator in which any single fault that can lead to a weight-2 data error always gives a nontrivial flag bit. Therefore, any possible fault combination with trivial flag bits from each weight-4 generator consists of at least two faults. Such a fault combination give rise to a single- or two-qubit data error (up to a multiplication of the same generator), thus the CNOT ordering for any weight-4 generator cannot impact the distinguishability of the fault set. That is, to search through all possible orderings of data CNOTs, it is sufficient to consider only $8!=40320$ CNOT orderings on the weight-8 generator.

We also consider all possible configurations of the two flag CNOT gates for the circuit measuring the weight-8 generator. Here, we assign a pair $(s,e)$ where $s,e\in \{0,1,\dots,8\}$ to the flag CNOTs, denoting two data CNOTs the two flag CNOTs are to be inserted after ($0$ refers to the location before the first data CNOT). For example, the configuration of the traditional flag circuit as in \cref{fig:flag_w} corresponds to the pair $(1,7)$. We iterate through all pairs of flag CNOTs with $s<e$, from $(0,1), (0,2),\dots$ to $(7,8)$.

Next, we generate two sets of possible logical operators of weight 5, the W2L and the W4L sets, which contain logical operators that overlap with the weight-8 generators on exactly two and four qubits, respectively. More precisely, let $Q$ be the set of supporting qubits of the weight-8 generator. The W2L set contains all logical operators $L$ such that $\mathrm{wt}(L)=5$ and $|\mathrm{supp}(L)\cap Q|=2$, and the W4L set that contains all logical operators $L$ such that $\mathrm{wt}(L)=5$ and $|\mathrm{supp}(L)\cap Q|=4$. Logical operators in the W2L and the W4L sets are shown in \cref{fig:err_d5_w2_data,fig:err_d5_w4_data}, respectively (for the W4L set, we only display the logical operators in which $\mathrm{supp}(L)\cap Q$ are distinct).  

\begin{figure*}[htbp]
	\centering
	\includegraphics[width=0.98\textwidth]{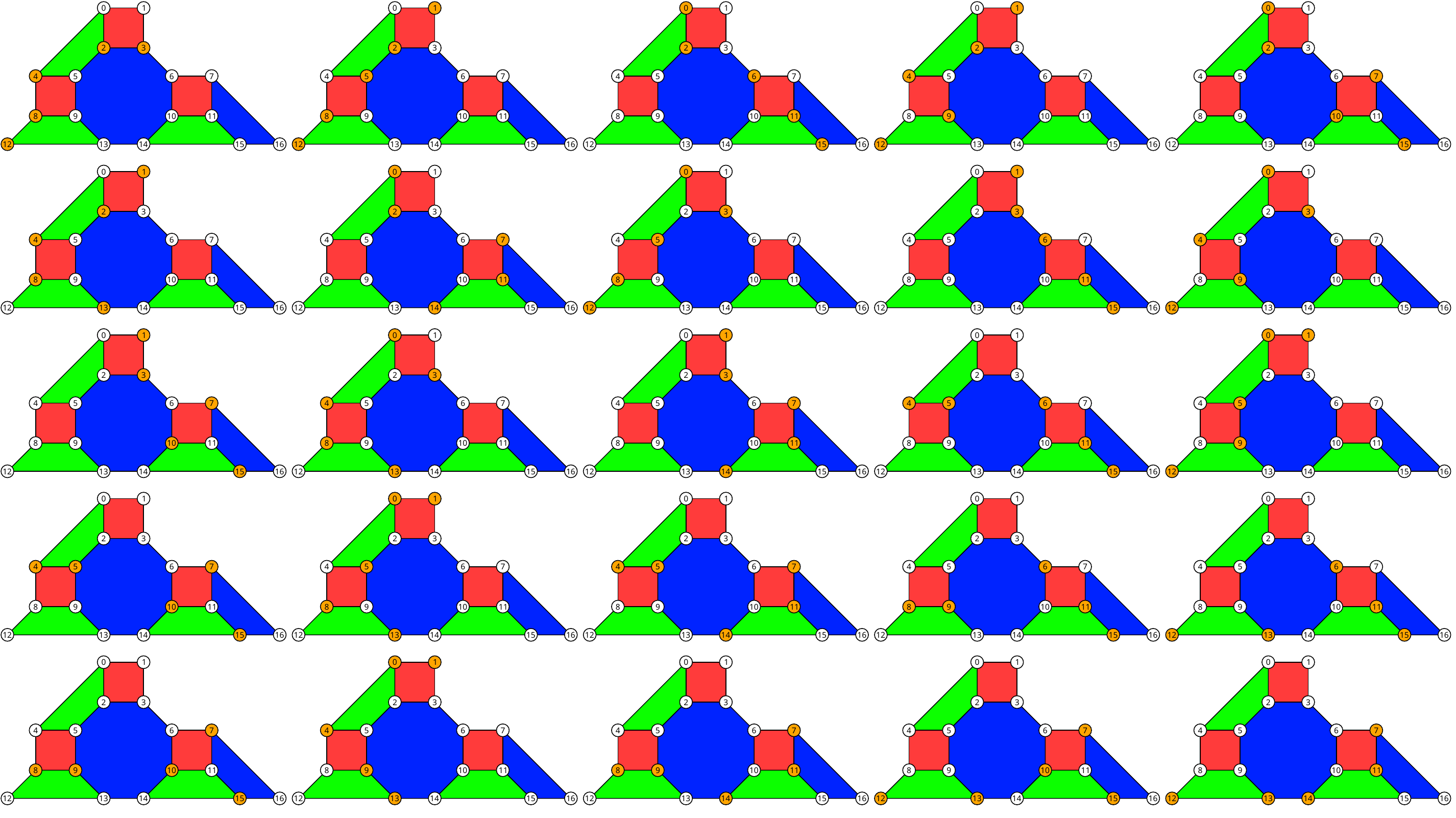}
	\caption{Weight-5 logical operators of the 4.8.8 color code of distance 5 that overlap with exactly two supporting qubits of the weight-8 stabilizer generator. Only logical operators with distinct sets of overlapping qubits are displayed.}
	%\caption{\HR{Length-5 logical strings in the distance 5 4.8.8 color code that have weight 2 overlap with the weight-8 stabilizer generator.}}
	\label{fig:err_d5_w2_data}
\end{figure*}

\begin{figure*}[htbp]
	\centering
	\includegraphics[width=0.98\textwidth]{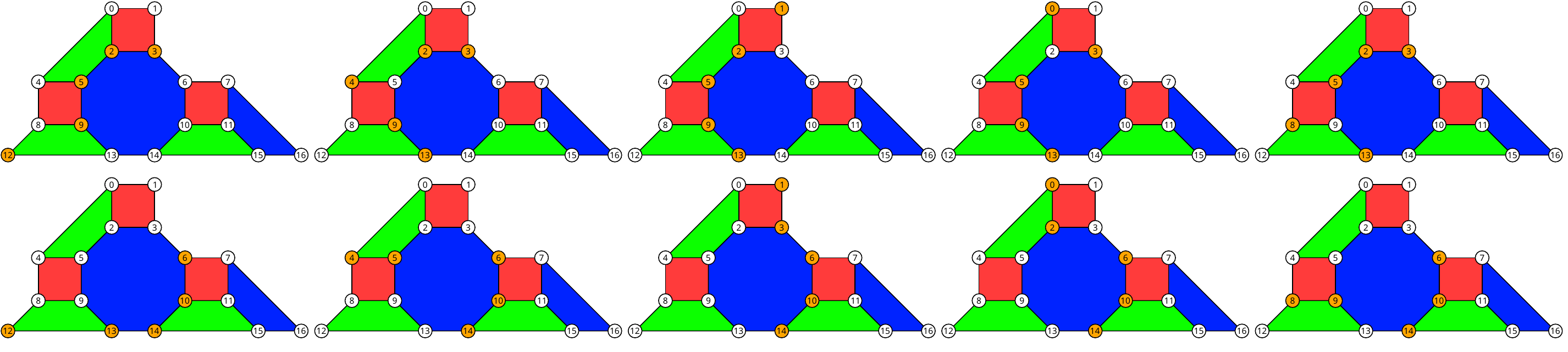}
	\caption{Weight-5 logical operators of the 4.8.8 color code of distance 5 that overlap with exactly four supporting qubits of the weight-8 stabilizer generator. The logical operators in the upper row are equivalent to the logical operators in the lower row up to a multiplication of the weight-8 stabilizer generator. }
	%\caption{\HR{The 5 pairs of possible length-5 logical strings in the distance 5 4.8.8 color code that have weight 4 overlap with the weight-8 stabilizer generator. The strings in the upper row are equivalent to the strings in the lower row up to multiplication with the weight-8 stabilizer generator.}}
	\label{fig:err_d5_w4_data}
\end{figure*}

\begin{figure}[tbp!]
	\centering
	\includegraphics[width=0.4\textwidth]{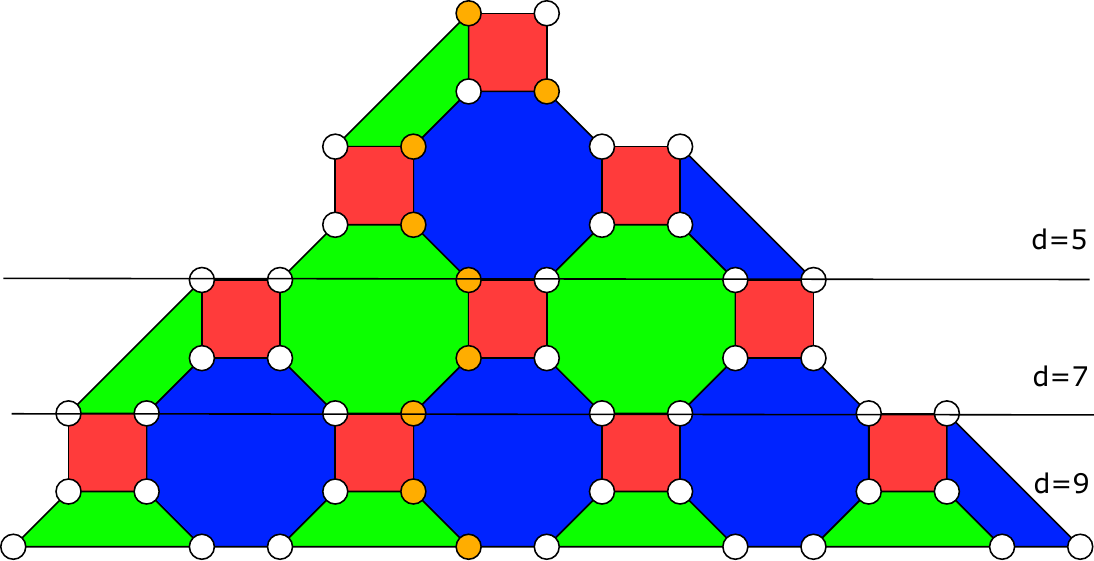}
	\caption{A weight-5 logical error on the 4.8.8 color code of distance 5 can be extended to a logical error of weight $5+2i$ on the code of distance $5+2i$ by adding $2i$ errors on the qubits that connect the old and the new bottom boundaries (depicted by horizontal lines).}
	\label{fig: app 488d975}
\end{figure}

For a given CNOT ordering of the weight-8 generator and a flag CNOT pair, the effective distance decreases if there exists a fault combination of 4 or fewer faults that results in one of the logical operators in the W2L or the W4L set. There are two cases that can cause the distance loss: (a) the case that a single, unflagged fault on the weight-8 generator leads to a weight-2 error whose support is the same as $\mathrm{supp}(L)\cap Q$ of some logical operator $L$ in the W2L set. This weight-2 error and 3 other faults from weight-4 generators can cause a logical operator in the W2L set. (b) the case that two faults on the weight-8 generator lead to an error of weight 3 (or 4) whose support is a subset of $\mathrm{supp}(L)\cap Q$ of some logical operator $L$ in the W4L set. This error of weight 3 (or 4) and 2 (or 1) other faults from weight-4 generators can cause a logical operator in the W4L set. In any case, 4 or fewer faults can cause a logical operator of weight 5. 

We tested the full population of $8!=40 320$ CNOT orderings for the weight-8 generator across all possible flag CNOT pairs and found that the FTEC protocol loses distance for all orderings. This proves that it is impossible to construct distance-preserving flag FTEC with single flag ancilla for the 4.8.8 color code of distance 5.

In Step (2), we observe that any of the logical operators in the W2L or the W4L set can be extended to a logical operator of weight $5+2i$ on any 4.8.8 color code of distance $5+2i, i=1,2,3,...$. This is due to the fact that each logical operator has support on only one qubit on the bottom boundary (the base of the triangle) of the code of distance 5. On the code of distance 7, an error of the same form can be connected to the red boundary by adding two data errors. We can repeat this process $i$ times for the code of distance $5+2i$. See \cref{fig: app 488d975} for an example.

Each CNOT ordering for the code of distance $5+2i$ has its corresponding CNOT ordering on the code of distance 5. Because for the code of distance 5, a logical operator in the W2L or the W4L set with trivial flag bits can arise from $4$ or fewer faults for all CNOT orderings, we find that for the code of distance $5+2i$, a logical operator of weight $5+2i$ with trivial flag bits can also arise from $4$ or fewer faults plus $2i$ data errors for all CNOT orderings. That is, distance-preserving flag FTEC with single flag ancilla is impossible for any 4.8.8 color code of distance $d \geq 5$.

%\newpage

%\bibliographystyle{ieeetr}
%\bibliographystyle{apsrev4-2}
\bibliography{bibtex_FT_Duke}

\end{document}